\documentclass[journal,11pt,draftclsnofoot,onecolumn]{IEEEtran}
\ifCLASSINFOpdf
\else
\fi
\usepackage{cite}
\usepackage{graphicx}
\usepackage{epsfig}
\usepackage{subfigure}
\usepackage{amsmath,amssymb}
\usepackage{amsthm}
\usepackage{slashbox}
\usepackage{standalone}

\newtheorem{thm}{Theorem}
\newtheorem{prop}{Proposition}

\newcommand{\beq} {\begin{equation}}
\newcommand{\eeq} {\end{equation}}
\newcommand{\beqr} {\begin{eqnarray*}}
\newcommand{\eeqr} {\end{eqnarray*}}

\begin{document}
\title{An Iterated Game of Uncoordinated Sharing of Licensed Spectrum Using Zero-Determinant Strategies}

\author{
    \IEEEauthorblockN{Ashraf Al Daoud\IEEEauthorrefmark{1}, George Kesidis\IEEEauthorrefmark{2}, J\"{o}rg Liebeherr\IEEEauthorrefmark{1} }\\
    \IEEEauthorblockA{\IEEEauthorrefmark{1}Department of Electrical and Computer Engineering,
    University of Toronto, Canada} \\
    \IEEEauthorblockA{\IEEEauthorrefmark{2}Department of Computer Science and Engineering,
    Pennsylvania State University, PA, USA}
}



\maketitle

\begin{abstract}
\boldmath
We consider private commons for secondary sharing of licensed spectrum bands with no access coordination provided by the primary license holder. In such environments, heterogeneity in demand patterns of the secondary users can lead to constant changes in the interference levels, and thus can be a source of volatility to the utilities of the users. In this paper, we consider secondary users to be service providers that provide downlink services. We formulate the spectrum sharing problem as a non-cooperative iterated game of power control where service providers change their power levels to fix their long-term average rates at utility-maximizing values. First, we show that in any iterated $2\times 2$ game, the structure of the single-stage game dictates the degree of control that a service provider can exert on the long-term outcome of the game. Then we show that if service providers use binary actions either to access or not to access the channel at any round of the game, then the long-term rate can be fixed regardless of the strategy of the opponent. We identify these rates and show that they can
be achieved using mixed Markovian strategies that will be clearly identified in the paper.
\end{abstract}

\IEEEpeerreviewmaketitle

\section{Introduction}
Advancements in mobile broadband access technologies in the recent years have resulted in an exponential surge in demand for wireless data services~\cite{Cisco}. While the demand is expected to grow, the wireless industry is trying to develop techniques to improve utilization of available radio spectrum bands. This has led to the advent of cognitive radio technologies which allow network users to adapt their system parameters to the dynamic environment and optimize spectrum utilization without necessarily cooperating to pursue such goals. Under this paradigm where centralization and traditional spectrum sharing techniques are no longer valid concepts in modern spectrum utilization, game theory has emerged as a useful tool for modeling and analyzing user behavior in non-cooperative spectrum sharing environments. See~\cite{Alpcan_book} for a recent survey on games for cognitive radio networks.

Non-cooperative spectrum sharing can be studied for scenarios that involve secondary provision of licensed spectrum where primary license holders lease out the surplus of their spectral capacities to some secondary users. These scenarios for spectrum sharing have been made possible by recent regulatory models that emerged under different proposals including the FCC's \emph{private commons} model~\cite{FCC} and the \emph{licensed shared access} model of the EU~\cite{RSPG}. Suggested models entail that spectrum provision may not necessarily require license holders to coordinate spectrum access among secondary users. For example, under the private commons model, secondary users are granted spectrum access via using peer-to-peer communications without relying on the license holder's infrastructure. In fact, license holders may not even need to have a deployed network in order to be eligible for this model~\cite{Buddhikot}. However, license holders can still authorize the use of certain communication devices or can dictate using specific technical parameters.

This paper presents a framework for designing strategies for secondary sharing of licensed spectrum bands. The underlying communication system involves a number of service providers that share an interference channel to provide downlink services without access coordination.  A distinctive feature of sharing licensed channels, versus sharing unlicensed ones, is that the utilities of the service providers from achieving some rate on the channel are discounted by the cost of utilizing the channel. This cost is paid to the primary license holder on usage basis in the form of monetary compensation. Since it is plausible to assume that the marginal utility of each service provider decreases by increasing the rate, transmitting at the maximum allowed power level can be sub-optimal from a utility maximization point of view. In this respect, operating at a utility-maximizing rate from the standpoint of any of the service providers is governed by the interference in the channel which depends on the demand patterns of the service providers. Specifically, at times when the demand is high, a service provider transmits at relatively higher power levels causing more interference, and visa versa when the demand is low, thus leading to variations in the interference. Demand patterns are generally unknown, and the key problem is to design strategies for power control that help service providers achieve their optimal rates and cope with fluctuations in the interference.


%

In this paper, secondary sharing of licensed spectrum is formulated as an iterative game of power control. Namely, at each round of the game service providers choose their transmission power levels and consequently achieve some downlink rates that depend also on the interference from other service providers. We show that, there exist strategies that allow service providers to fix their rates, on the long term, regardless of the strategies of their opponents. The key idea is to realize that in iterated games with same action space and same payoff profiles, players with longer memories have no advantage over players with shorter ones. Therefore, players can in any round condition their moves on the outcome of the game in the previous round. This implies that iterated games can lend themselves to Markovian analysis where a player's strategy can be defined in terms of the state transition probabilities of the resulted Markovian chain.

We use this insight in the power control game. First, we show that in any $2\times2$ iterated game,\footnote{Two players each with boolean action space leading to a $2\times 2$ payoff matrix.} the structure of the payoff matrix of any of the players in the single-stage game dictates whether or not a player can control its long-term payoff of the game. We also show that this property can be realized in the power control game by transforming the action space of the service providers into a binary space; either to access or not to access the channel in each round of the game. The approach provides the players with full control on a range of rates that will be clearly characterized in the paper. In essence, a player can achieve any value in the valid range by iterating its actions using mixed (probabilistic) Markovian strategies. The intuition behind this approach is to allow players to maintain a certain rate, on the long term, by using reactive strategies such that whenever the average rate exceeds the targeted value, it can be lowered by not participating in the channel in some future rounds. The paper identifies these strategies and shows that any fixed outcome of the game can be achieved using multiple strategies that differ by their convergence rates.

Game theory for sharing wireless resources is a widely studied topic in communication networks, in particular, using
the classical theory of non-cooperative games~\cite{Osborne}. Proposed strategic-game models for sharing interference channels
have included pricing mechanisms~\cite{Mandayam, Huang, Niyato, Attar}, medium access control~\cite{Jin,Jin2}, and transmission power control of both the uplink (end-users to base-stations) and the downlink (base-stations to end-users)~\cite{Alpcan, Altman, Chung}. A common assumptions in these models is that game structure and rationality of players are common knowledge in the game, i.e., players hold beliefs about each others' strategic choices~\cite{Aumann}. Such assumptions are not limited to single-stage games but also extend to iterated games where players interact in multiple rounds.

Iterated games are studied to induce cooperation in self organizing wireless {\em ad-hoc} networks. Most of the studies use the iterated Prisoners' Dilemma game to model packet forwarding between nodes~\cite{Korilis, Felegyhazi, Jaramillo}. The model is motivated by many experimental studies which show that Tit-For-Tat can be an efficient strategy in this regard~\cite{Axelrod}. An important feature of iterated games is the fact that an action taken by a player at any round of the game has an impact on the future actions of the other players. This in turn leads to the concept of punishment for deviating from equilibrium strategies. Such techniques are applied in the area of sharing unlicensed bands, particularly in~\cite{Etkin}, where multiple systems coexist and interfere with each other. In essence, in~\cite{Etkin}, spectrum sharing is modeled as an iterated power control game to devise self-enforcing power control rules that lead to fair and efficient Nash equilibria.

In this work, we follow a different approach from~\cite{Etkin} and seek power control strategies that allow service providers to share spectrum on the downlinks and
maintain average 
rates that are robust to variations in power control strategies of the opponents.
Our work is motivated by recent results in the theory of
iterated Prisoners' Dilemma games, where Press and Dyson have recently shown in~\cite{Press}
that these games admit \textit{zero-determinant strategies}, which in some cases allow players to control each others' long-term payoffs, and in other cases allow them to set a linear relationship between the payoffs. Such strategies are realized if we observe that progression of the game can be formulated as a one-step Markov process. The approach of~\cite{Press}, from which the earlier results of \cite{Boerlijst} can be obtained as a special case, more readily admits generalization to asymmetric games, with different payoff structures that can involve more than two players.

The major contribution of our paper can be summarized as follows:
\begin{itemize}
\item[1)] We present an extension of the approach in~\cite{Press} and clearly identify structures of $2\times2$
games that allow player to control its own payoff or that of the opponent, thus implying a broader application that is not restricted to the Prisoners' Dilemma game. We follow a generalized approach as we do not assume symmetric payoffs of the players, and thus no assumption of symmetric control on the outcome of the game.
\item[2)] We identify the feasible set of payoffs that a player can guarantee from the game and identify mixed Markovian strategies (zero-determinant strategies) for each possible outcome of the game. Furthermore, we show that the notion of payoff control is not restricted to two-player games, but can be extended to games with multiple players.
\item[3)] We formulate secondary sharing of wireless spectrum as an iterated game of power control. We use an economic model for downlink data transmission to argue that, in interference channels  with no access coordination, players can fix their long-term rates at utility-maximizing values by taking binary actions (e.g., either to transmit at maximum power or not to transmit). In this regard, we identify strategies for iterating these actions and study their convergence.
\end{itemize}

The rest of the paper is organized as follows: In Section~\ref{sec:theory}, we address zero-determinant strategies for $2\times2$ iterated games and present our results for games of general payoff structures. We also extend our results to include games with more than two players. In Section~\ref{sec:sharing}, we analyze secondary spectrum sharing as an iterated game of power control and devise strategies for the proposed $2\times2$ game. A numerical study to analyze convergence and power consumption of these strategies is provided in Section~\ref{sec:numerics}. Finally, the paper concludes in Section~\ref{sec:concl}.

\section{Zero-Determinant Strategies for Iterated Games} \label{sec:theory}

In this section, we develop new results on iterated games where the action space and the payoff matrix do not change over the course of the game. Our analysis is based on the approach of \cite{Press} which shows that there exist strategies for indefinitely iterated  $2\times2$ games that are referred to as ``zero-determinant" strategies and which allow the players to control their long-term payoffs or the payoffs of their opponents. The type of control that a player can exert hinges on the structure of the game. In this section, we identify these structures and any feasible set of payoffs that can be controlled. We also identify the strategies that lead to this control.

For this purpose, consider a $2\times2$ iterated game with the single round payoff matrix given in Figure~\ref{fig:iterated}. In each round of the game, row player $X$ and column player $Y$ have binary actions, respectively $n_1,n_2\in \{1,2\}$, leading to payoffs, respectively  $X_{\bf n},Y_{\bf n}$ where ${\bf n} = (n_1,n_2)$. A salient feature of iterated games is that players with longer memories of the history of the game have no advantage over those with shorter ones, i.e., a strategy of a player that shares the same history used by the opponent does not gain more from using longer history of the game. This is due to the iterative nature of the game where actions and payoffs are indefinitely fixed (see the Appendix of \cite{Press}), and thus, strategies can be designed by assuming that the players have memories of a single move.


\begin{figure}[]
\centering
\begin{tabular}{c|c|c}
\backslashbox{Player $X$}{Player $Y$} & $n_2=1$ & $n_2=2$ \\ \hline
$n_1=1$ & $(X_{1,1},Y_{1,1})$ & $(X_{1,2},Y_{1,2})$  \\ \hline
$n_1=2$ & $(X_{2,1},Y_{2,1})$ & $(X_{2,2},Y_{2,2})$
\end{tabular}
\caption{One round payoff matrix of the iterated two-player two-action game.}
\label{fig:iterated}
\end{figure}

\subsection{Zero Determinant Strategies for $2\times2$ Games} \label{strategies}
We describe the state of the game in any round by the actions of the players in that round. Specifically, let $\Omega$ denote the set of all possible states,\footnote{In the application to spectrum sharing which will be addressed later in the paper,
play action $1$ will correspond to accessing the channel with maximum power, while play action $2$ will correspond to accessing the channel with lower power or not accessing the channel at all.} i.e.,
\beq
\Omega=\{(1,1), (1,2), (2, 1), (2,2)\},
\label{eq.states}
\eeq
and let $\mathbf{n}(t)$ denote the state of the game in round $t \geq 0$.
In each round, players choose their actions with probabilities that depend
on the state of the game in the previous round and thus the process
$\{\mathbf{n}(t): t=0,1,\cdots\}$  can be modeled as a Markov chain.

In this respect, consider player~$X$ and let
\[
	p^{\bf k}_1 = {\sf Pr}\left(n_1(t+1)=1~|~ {\bf n}(t)={\bf k}\right), ~~\forall {\bf k}\in\Omega,
\]
denote the probability that player~$X$ takes action~$1$ in round $t+1$ if in the previous round, player $X$ took action  $k_1$
and player $Y$ took action $k_2$.
For player $Y$, similarly let
\[
	p^{\bf k}_2 = {\sf Pr}(n_2(t+1)=1~|~ {\bf n}(t)={\bf k}), ~~\forall {\bf k}\in\Omega.
\]
To simplify notation, we write
\[
p^{\bf k} = p^{\bf k}_1  ~~\mbox{and} ~~q^{\bf k} = p^{\bf k}_2.
\]
The set of actions of a player is referred to as the strategy of that player, i.e., $\{p^{\bf k},\forall {\bf k}\in\Omega\}$ is a strategy of player $X$ and $\{q^{\bf k},\forall {\bf k}\in\Omega\}$ is a strategy of player $Y$. The state transition matrix of the Markov chain can be described as follows assuming that the rows and the columns are in the same order as listed
in~(\ref{eq.states}):
\[ \begin{array}{ll}
  \mathbf{M}=   &
  \mbox{
$\begin{pmatrix}
    p^{1,1}q^{1,1} & p^{1,1}(1-q^{1,1}) & (1-p^{1,1})q^{1,1} & (1-p^{1,1})(1-q^{1,1})\\
p^{1,2}q^{2,1} & p^{1,2}(1-q^{2,1}) & (1-p^{1,2})q^{2,1} & (1-p^{1,2})(1-q^{2,1}) \\
p^{2,1}q^{1,2} & p^{2,1}(1-q^{1,2}) & (1-p^{2,1})q^{1,2} & (1-p^{2,1})(1-q^{1,2})\\
p^{2,2}q^{2,2} & p^{2,2}(1-q^{2,2}) & (1-p^{2,2})q^{2,2} & (1-p^{2,2})(1-q^{2,2})
  \end{pmatrix}$.}
\end{array} \]

Let $\pi_{i,j}$ be the probability that player $X$ takes action $i$ and player $Y$ takes action $j$. The Markov chain has a stationary distribution
 $\mbox{\boldmath${\pi}$}^{\rm T}=(\pi_{1,1},\pi_{1,2},\pi_{2,1},\pi_{2,2})$ that satisfies
\[
	\mbox{\boldmath${\pi}$}^{\rm T}\mathbf{M}=\mbox{\boldmath${\pi}$}^{\rm T}.
\]
$\mbox{\boldmath${\pi}$}$ is unique if and only if the chain has a unique closed communication class. In this case, the long-term payoff for player $X$ is given by
\beq
u_X= \mbox{\boldmath${\pi}$}^{\rm T}\mathbf{X}
\label{eq.revX}
\eeq
and for player $Y$ is given by
\beq
u_Y= \mbox{\boldmath${\pi}$}^{\rm T}\mathbf{Y},
\label{eq.revY}
\eeq
where
 \[
\mathbf{X} =  \left( \begin{array}{c}
 X_{1,1}\\
 X_{1,2}\\
 X_{2,1} \\
 X_{2,2} \end{array} \right)
 ~~~~~\mbox{and} ~~~~~
   \mathbf{Y} =  \left( \begin{array}{c}
Y_{1,1} \\
Y_{1,2} \\
Y_{2,1}  \\
Y_{2,2}   \end{array} \right).
\]
Assuming unique stationary $\mbox{\boldmath${\pi}$}$, we can write
\[
\mbox{\boldmath${\pi}$}^{\rm T}\mathbf{\tilde{M}}=\mathbf{0},
\]
where $\mathbf{\tilde{M}}= \mathbf{M}-\mathbf{I}$.
Let ${\sf adj}(\mathbf{\tilde{M}})$ be the adjugate matrix,
i.e., the  transposed matrix of signed minors.
By Cramer's rule,
\[
{\sf adj}(\mathbf{\tilde{M}})\mathbf{\tilde{M}}=
{\sf det}(\mathbf{\tilde{M}})I=0,
\]
where the second equality holds because $\mathbf{\tilde{M}}$ is singular.
This implies that each row of the matrix ${\sf adj}(\mathbf{\tilde{M}})$ is proportional to $\mbox{\boldmath${\pi}$}$. Furthermore,
for an arbitrary vector $\mathbf{f}$,
$\mbox{\boldmath${\pi}$}^{\rm T} \mathbf{f}$
is the determinant of a modified version of $\mathbf{\tilde{M}}$
with a column replaced by $\mathbf{f}$.
The determinant does not change by adding the first column of
this modified matrix to the second and third columns, so that
 \beq
 \mbox{\boldmath${\pi}$}^{\rm T}\mathbf{f} =
 {\sf det} \left( \begin{array}{cccc}
-1+p^{1,1}q^{1,1} & -1+p^{1,1} & -1+q^{1,1} & f_1\\
p^{1,2}q^{2,1} & -1+p^{1,2} & q^{2,1} & f_2 \\
p^{2,1}q^{1,2} & p^{2,1} & -1+ q^{1,2} & f_3\\
p^{2,2}q^{2,2} & p^{2,2} & q^{2,2} & f_4 \end{array} \right). \label{eq.det}\eeq

A key observation of \cite{Press}
is that the second and the third columns of the matrix in~(\ref{eq.det})
are purely dependent on the actions of player $X$ and player $Y$, respectively.
In specific,
 \[
{\mathbf{\tilde{m}}}_X=  \left( \begin{array}{c}
-1+p^{1,1} \\
-1+p^{1,2}  \\
p^{2,1}  \\
p^{2,2}  \end{array} \right)
 ~~~~~\text{and} ~~~~~
 {\mathbf{\tilde{m}}}_Y =  \left( \begin{array}{c}
-1+q^{1,1} \\
q^{2,1} \\
-1+ q^{1,2}  \\
q^{2,2}    \end{array} \right).
\]
Without loss of generality, consider the game from the standpoint of player $X$.
If ${\mathbf{\tilde{m}}}_X = \mathbf{f}$, the determinant in~(\ref{eq.det}) is equal to $0$, and thus
if
$\mathbf{f}= a\mathbf{X}+ b$, Equation~(\ref{eq.det})
is
\beq
a u_X + b = 0, \label{Press-theorem}
\eeq
where $u_X$ is defined in~(\ref{eq.revX}), and $a$ and $b$ are non-zero real numbers.

Player $X$ can thus fix the value of $u_X$ regardless of the strategy of player~$Y$. To achieve this, the values of $a$ and $b$ should be chosen such that  $p^{1,1}, p^{1.2},p^{2,1},$ and $p^{2,2}$ are probabilities which in turn depends on the structure of the game via the equality ${\mathbf{\tilde{m}}}_X = a\mathbf{X} + b$.

In the following theorem, we state our first result that defines the structures of~$2\times 2$ games where player~$X$ can control $u_X$ and defines the strategies that lead to such control.


\begin{thm} \label{thm:fixing}
For $k=1,2$, let $X_{k,\min}$ and $X_{k,\max}$, respectively, denote the minimum and maximum value of row $k$ in the payoff matrix of a $2\times2$ iterated game. Specifically,
\begin{eqnarray}
X_{k,\min} &=& \min(X_{k,1}, X_{k,2}), \nonumber  \\
X_{k,\max} &=& \max(X_{k,1}, X_{k,2}). \nonumber
\end{eqnarray} Player~$X$ can control its long-term payoff $u_X$ regardless of the action of player $Y$ if and only if there exist $k_{\max}, k_{\min} \in \{1,2\}$ where
\beq
X_{k_{\max},\max} \leq X_{k_{\min},\min}. \label{payoff_cond}
\eeq
If so, any value of $u_X$ from the interval $[X_{k_{\max},\max},
X_{k_{\min},\min}]$ can be achieved by using the following mixed/probabilistic
strategies:
\begin{eqnarray}
p^{1,1} &=& 1 + \left(1 -\frac{X_{1,1}}{u_X}\right)b, \label{eq.cond1} \\
p^{1,2} &=& 1 + \left(1 -\frac{X_{1,2}}{u_X}\right)b, \label{eq.cond2} \\
p^{2,1} &=& \left(1 -\frac{X_{2,1}}{u_X}\right)b, \label{eq.cond3} \\
p^{2,2} &=& \left(1 -\frac{X_{2,2}}{u_X}\right)b, \label{eq.cond4}
\end{eqnarray}
where  $b$ is chosen such that, if $k_{\min} = 1$ and $k_{\max} = 2$, then
\[
 0 < b \leq \min\left(\frac{-1}{1- \frac{X_{1,max}}{u_X}},\frac{1}{1- \frac{X_{2,min}}{u_X}}\right),
\]
and if $k_{\min} = 2$ and $k_{\max} = 1$, then
\[
 \max\left(\frac{-1}{1- \frac{X_{1,min}}{u_X}},\frac{1}{1- \frac{X_{2,max}}{u_X}}\right) \leq b < 0.
\]
\end{thm}
\begin{proof}
We need to obtain $a$ and $b$ that satisfy $\mathbf{\tilde{m}_X} = a\mathbf{X}+ b$ and render $p^{1,1}, p^{1,2}, p^{2,1}, \text{and}~p^{2,2}$ as probabilities. First note that, by~(\ref{Press-theorem}), $a = -\frac{b}{u_X}$ and thus formulae~(\ref{eq.cond1}--\ref{eq.cond4}) follow. Next, we obtain the range of valid values of the non-zero variable $b$ by dividing the search domain into two intervals; $b>0$ and $b<0$.

\subsection*{\underline{Case 1 $(b>0)$}:}
Consider ~(\ref{eq.cond1}) and~(\ref{eq.cond2}) and notice that, for any value of $b > 0$ and a given value of $u_X$, the condition
\[
X_{1,min} \geq u_X
\]
is a sufficient and necessary condition for $p^{1,1}$ and $p^{1,2}$ to be less than or equal to~$1$.
Similarly, for $p^{2,1}$ and $p^{2,2}$ to be greater than or equal to~$0$, we obtain the following condition from~(\ref{eq.cond3}) and~(\ref{eq.cond4})
\[
X_{2,max} \leq u_X.
\]
Therefore,  $u_X$ cannot be fixed at values outside the interval $[X_{2,max}, X_{1,min}]$. To show that $u_X$ can be fixed at any value in this interval, we need to show that there exist $b>0$ such that $p^{1,1}$ and $p^{1,2}$ are greater than or equal to $0$ and such that $p^{2,1}$ and $p^{2,2}$ are less than or equal to $1$. In this regard, from~(\ref{eq.cond1}) and~(\ref{eq.cond2}) we obtain
\begin{eqnarray}
 0 < b &\leq&  \frac{-1}{1- \frac{X_{1,1}}{u_X} }, \nonumber \\
 0 < b &\leq&  \frac{-1}{1- \frac{X_{1,2}}{u_X} }. \nonumber
\end{eqnarray}
Note that, since  $X_{1,1}, X_{1,2} > u_X$, the tightest upper bound is $\frac{-1}{1- \frac{X_{1,max}}{u_X} }$.
In the same way, from~(\ref{eq.cond3}) and~(\ref{eq.cond4}) we obtain
\begin{eqnarray}
 0 < b &\leq&  \frac{1}{1- \frac{X_{2,1}}{u_X} }, \nonumber \\
 0 < b &\leq&  \frac{1}{1- \frac{X_{2,2}}{u_X} }, \nonumber
\end{eqnarray}
and since  $  X_{2,1}, X_{2,2} < u_X $, $\frac{1}{1- \frac{X_{2,min}}{u_X} }$ is the tightest upper bound. Therefore,  $u_X$ can be fixed at any value in the interval $[X_{2,max}, X_{1,min}]$ by choosing $b$ from the following feasible range
\[ 0 < b \leq \min\left(\frac{-1}{1- \frac{X_{1,max}}{u_X}},\frac{1}{1- \frac{X_{2,min}}{u_X}}\right). \]
\subsection*{\underline{Case 2 $(b<0)$}:}
We can follow the same steps in the previous case. Namely, for $p^{1,1}$ and $p^{1,2}$ to be less than or equal to~$1$, it is required that
\[
X_{1,max} \leq u_X,
\]
and for  $p^{2,1}$ and $p^{2,2}$ to be greater than or equal to $0$, it is required that
\[
X_{2,min} \geq u_X.
\]
Combining the previous conditions yields the new condition
\[
X_{1,max} \leq X_{2,min},
\]
Furthermore, for $p^{1,1}$ and $p^{1,2}$ to be greater than or equal to~$0$, we obtain the conditions
\[  \frac{-1}{1- \frac{X_{1,min}}{u_X} }  \leq b < 0.  \]
In a similar fashion, for $p^{2,1}$ and $p^{2,2}$ to be less  than or equal to~$1$ we obtain
\[  \frac{1}{1- \frac{X_{2,max}}{u_X} }  \leq b < 0.  \]
Therefore, $b$ can be chosen from the following feasible range
\[ \max\left(\frac{-1}{1- \frac{X_{1,min}}{u_X}},\frac{1}{1- \frac{X_{2,max}}{u_X}}\right) \leq b < 0. \]
\end{proof}

%

Theorem~\ref{thm:fixing} provides a framework for understanding payoff control in $2\times2$ iterated games. It states that the structure of the payoff matrix reveals the possibility of players controlling their long-term payoffs. In fact, only if the maximum payoff in one row is less than or equal to the minimum in the other row, then row player~$X$ can set the long-term payoff, $u_X$, at any value between the minimum and the maximum, inclusive. For example if $X_{1,1}=1, X_{1,2}=0.75$, and $X_{2,1}=X_{2,2}=0.5$, then player~$X$ can set $u_X$ at any value in the interval $[0.5,0.75]$.

The results in the theorem can be directly applied to player $Y$ by considering the columns of the payoff matrix of the player instead of the rows. In particular, let
\begin{eqnarray*}
Y_{k,\min} = \min(Y_{1,k}, Y_{2,k})  & \mbox{and} &
Y_{k,\max} = \max(Y_{1,k}, Y_{2,k}),
\end{eqnarray*}
then player $Y$ can control its long-term payoff, $u_Y$, if and only if there exists $k_{\max}, k_{\min} \in \{1,2\}$, where $k_{\max} \neq k_{\min}$ and
$Y_{k_{\max},\max} \leq Y_{k_{\min},\min}$.

A simple example to verify strategies~(\ref{eq.cond1}--\ref{eq.cond4}) in the theorem is to choose $u_X = X_{k_{\min},\min}$, meaning that player $X$ sets $u_X$ at the maximum value possible. If we assume that $X_{k_{\min},\min}=X_{1,1}$, then regardless of the value of $b$, this always yields a strategy with $p^{1,1}=1$, i.e., player $X$ plays action~$1$ whenever both players played this action in the previous round. One way to understand this result is to consider a strategy of player $Y$ playing action~$1$ in each round of the game. Player $X$ will be then playing action~$1$ in each round as there will be no opportunity to make up for losses that may result in from not playing that action in any of the previous rounds.

In the same example, assume that the ratio $X_{1,1}/X_{1,2}=0.5$ and assume that $X_{2,1}=X_{2,2}=0$. If $b=1$, then this yields the deterministic strategy $p^{1,1}=1$, $p^{1,2}=0$, $p^{2,1}=1$, and $p^{2,2}=1$. This strategy is quite intuitive since it tracks the payoff of player $X$ such that, regardless of the strategy of player $Y$, whenever the payoff in any round is $X_{1,2}$, player $X$ plays action $2$ in the next round and gains $0$ payoff, so that the average of the two rounds is maintained at the targeted value $X_{1,1}$. In the next round, the player plays action $1$ to gain at least $X_{1,1}$ and so on. The strategy is one of several strategies that can be obtained by changing the value of the variable~$b$, and which will be discussed in more details in Section~\ref{sec:numerics}.

Two important observations can be carried from Theorem~\ref{thm:fixing}. First, the players can design their strategies without an underlying assumption of knowledge of each other's payoffs. All that a player needs to know about the opponent at any round is the latter's action in the previous round. This leads us to the second observation which highlights a more general perspective of this theorem. In essence, if the structure of the payoff matrix of the opponent admits the same structure described in the theorem, then a player can control the long-term payoff of the opponent. For example, player $X$ can control the payoff of player $Y$ if the payoff matrix of player $Y$ satisfies conditions~(\ref{payoff_cond}) with $X_{i,j}$ replaced by $Y_{i,j}$.

Controlling opponent's payoff can be realized in games such as the iterated Prisoners' Dilemma where $Y_{1,2} > Y_{1,1} > Y_{2,2} > Y_{2,1}$. In this game, the row player can set the payoff of the column player at any value in the interval $[Y_{1,1}, Y_{2,2}]$. Strategies for opponents controlling  each other's payoffs were previously studied in~\cite{Boerlijst} and presented for a subset of games where players have symmetric payoffs  as in the case of the Prisoners' Dilemma game.

\subsection{Iterated Games with Multiple Players }
The results presented in the previous section can be extended to include games of more than two players. Let $N \geq 2$ denote the number of players in the game and assume they are indexed $1,2,\cdots, N$. Let the binary vector $\mathbf{n}(t)=(n_i(t): i=1,\cdots, N)$ describe the state of the game in a given round $t$, where $n_i(t)\in\{1,2\}$ for all $i,t$
so that at any given $t$, $\mathbf{n}(t) \in \{1,2\}^N =:\Omega$. The process $\{\mathbf{n}(t): t=0,1,\cdots\}$ can be described as a multi-dimensional Markov chain. In each round of the game, players take actions with probabilities that depend on the state of the game in the previous round.

Let $p_{i}^{\mathbf{k}}$ denote the probability that player $i$  plays action $1$ in a certain round if the game was in state $\mathbf{k}$ in the previous round, and let
\[
\mathbf{p}_i=(p_{i}^{\mathbf{k}}: \mathbf{k}\in \Omega)
\]
denote the complete strategy profile of player $i$. The state transition matrix of the $N$-player game can be presented as a $2^N\times 2^N$ matrix.
Similar to the game with $N=2$ players, we can apply Cramer's rule to
$\mathbf{\tilde{M}}= \mathbf{M}-\mathbf{I}$ and replace  the last
column (corresponding to all players taking action $2$)
with a ``reward" vector {\bf f}.
For ${\bf l,k}\in \Omega$,  the entry in the
${\bf k}^{\rm th}$ row and ${\bf l}^{\rm th}$ column  of
$\mathbf{\tilde{M}}$  is for all iterations $t\geq 0$
\[
{\sf Pr}\left({\bf n}(t+1)={\bf k} ~|~ {\bf n}(t)={\bf l}\right)  =
\prod_{i\in \mathcal{K}^{\mathbf{k}}}
p_i^{\mathbf{l}}
\prod_{j \in \mathcal{L}^{\mathbf{k}}}
(1-p_j^{\mathbf{l}})
\]
where $\mathcal{K}^{\mathbf{k}}$ is the set of players playing action $1$
in state $\mathbf{k}$ and $\mathcal{L}^{\mathbf{k}}$ is the set of
players playing action $2$ in state $\mathbf{k}$.

Consider adding all columns $\mathcal{C}_i\subset \Omega$ of $\tilde{{\bf M}}$ which correspond
to states where player $i$ and at least one other player plays action $1$,
to the column where only player $i$ plays action $1$. An entry of the resulting column ${\bf \tilde{m}}_i$
at row $\mathbf{k}$ is then given as
\[
\left\{\begin{array}{ll}
-1 + p_i^{\mathbf{k}}\Gamma  &\mbox{if a diagonal element of
$\mathbf{\tilde{M}}$ is added to this entry,}\\
p_i^{\mathbf{k}}\Gamma  &\mbox{otherwise,}
\end{array}\right.
\]
where

\[
\Gamma = \sum_{{\bf k}\in\mathcal{C}_i} \prod_{j\in\mathcal{K}^{{\bf k}}\backslash i} p_j^{{\bf k}}
\prod_{l\in\mathcal{L}^{{\bf k}}} (1-p_l^{{\bf k}}).
\]

An important observation is that $\Gamma=1$ since each product in $\Gamma$ has elements that are either the probability or its complement of a fixed set of events. Thus, the sum of all possible permutations of these products evaluate to 1. Therefore, in a similar reasoning that led to~(\ref{eq.det}), multiplying the stationary distribution of the game \mbox{\boldmath${\pi}$} with an arbitrary $|\Omega|$-size vector $\mathbf{f}$ leads to the following structure (also displaying ${\bf \tilde{m}}_i$)
\beq
{\mbox{\boldmath${\pi}$}}^{\rm T} \mathbf{f} =
{\sf det} \left( \begin{array}{cccc}
\cdots & -1+p_i^{\mathbf{k}} & \cdots & f_1\\
\cdots & -1+p_i^{\mathbf{k}} & \cdots & f_2 \\
\ddots& \vdots & \ddots & \vdots\\
\cdots & p_i^{\mathbf{k}} & \cdots & f_{|\Omega|-1} \\
\cdots & p_i^{\mathbf{k}}& \cdots & f_{|\Omega|} \end{array} \right). \nonumber
\eeq
So, a column that corresponds to the state where only player $i$
plays action $1$ has elements that depend solely on the actions of that player.

We follow the developments that led to Theorem~\ref{thm:fixing} and let $U_{i,\mathbf{n}}$ denote the payoff of player $i$ if the state of the game at the previous round was $\mathbf{n}$. We also let $\mathbf{U}_i$ denote a vector of all possible outputs. Let $u_i$ denote a generic value of the long-term payoff of player $i$. Thus, taking ${\bf \tilde{m}}_i = {\bf f}= a_i\mathbf{U}_i + b_i$, where $a_i$ and $b_i$ are non-zero real numbers leads to zero-determinant strategies for own payoff control, which is formulated in the following proposition:

\begin{prop} \label{prop:multi}
In the game with $N\geq 2$ players, for k=1,2, let
\begin{eqnarray}
U_{i,k,\min} &=& \min(U_{i,\mathbf{n}}: n_i=k), \nonumber  \\
U_{i,k,\max} &=& \max(U_{i,\mathbf{n}}: n_i=k)), \nonumber
\end{eqnarray}
where the first quantity is the minimum payoff of player $i$ when playing action $k$, and the second quantity is the maximum.

Player $i$ can control its long-term payoff, $u_i$, regardless of the actions of the other players in the game
if and only if there exists $k_{\max}, k_{\min} \in \{1,2\}$ where
\beq
U_{i,k_{\max},\max} \leq U_{i,k_{\min},\min}. \nonumber
\eeq
If so, any value of $u_i$ from the interval $[U_{i,k_{\max},\max},
X_{i,k_{\min},\min}]$ can be achieved by using the following
strategies:
\[
p_i^{\mathbf{k}}=  \left\{\begin{array}{ll}
 1 + (1-\frac{U_{i,\mathbf{k}}}{u_i})b_i, &\mbox{if player $i$ plays action $1$ in state $\mathbf{k}$,} \\
(1-\frac{U_{i,\mathbf{k}}}{u_i})b_i,&\mbox{otherwise,}
\end{array}\right.
\]
where  $b_i$ is chosen such that, \newline if $k_{\min} = 1$ and $k_{\max} = 2$, then
\[
 0 < b_i \leq \min\left(\frac{-1}{1- \frac{U_{i,1,max}}{u_i}},\frac{1}{1- \frac{U_{i,2,min}}{u_i}}\right),
\]
and if $k_{\min} = 2$ and $k_{\max} = 1$, then
\[
 \max\left(\frac{-1}{1- \frac{U_{i,1,min}}{u_i}},\frac{1}{1- \frac{U_{i,2,max}}{u_i}}\right) \leq b_i < 0.
\]
\end{prop}

\section{A Non-Cooperative Game for Sharing Licensed Spectrum} \label{sec:sharing}
In this section, we apply our results to design strategies for sharing licensed spectrum bands. We consider a general model for spectrum sharing that involves $N$ service providers indexed $i=1,2,\cdots, N$ sharing a channel of bandwidth $W$. We assume that the channel is primarily licensed to a single entity that we refer to as the license holder. We consider a cold leasing model where the license holder may not deploy any network or equipment, and thus, offering the channel to the service providers without access coordination. This model is one of several models that have been suggested for treating spectrum bands as private commons where the ultimate ownership of spectrum is preserved by the license holder (See for example~\cite{Buddhikot}).

We model the underlying communication system as an interference channel where at times when the channel is less congested, the service providers create less interfere to each other, and thus can achieve better throughput rates. We focus on the downlink and assume that the service providers have fixed pools of end-users co-located within a certain geographical area. See Figure~\ref{fig:model} for a description of this model.  Let $\mathcal{S}_i$ denote the set of end-users of service provider $i$. The license holder regulates channel access by imposing a limit on the maximum transmission power of each service provider. It also allocates the underlying code space for transmission to individual end-users. Power and code allocations are normally negotiated with the license holder and provided through ``secondary provider" contracts.

\begin{figure}[]
\centering
\includegraphics[width=3.5in]{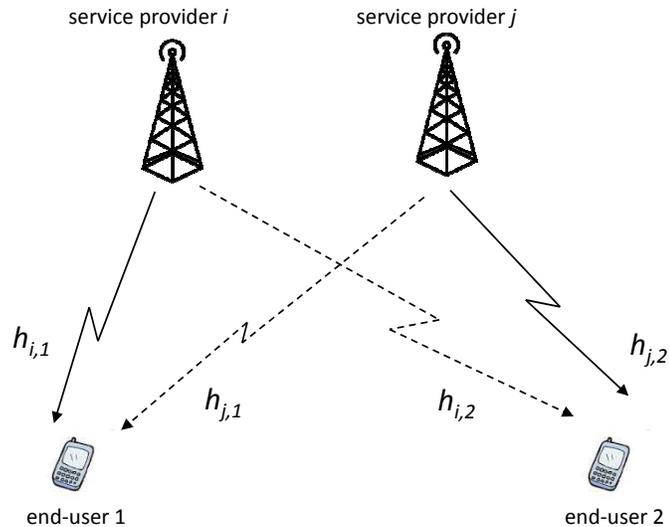} \vspace*{-.5cm}
\caption{A channel access model where a number of service providers share a channel and provide downlink services to groups of end-users spatially located in over-lapping coverage areas. Each end-user receives service for its designated service provider, but also gets interference from other service providers.}
\label{fig:model}
\end{figure}

We follow a simple model of common-channel interference under CDMA where transmission of a service provider to a given end-user appears as noise
to all other end-users, including those belonging to other service providers. While interference cancellation techniques can be still applied, they are precluded in this model due to practical limitations such as decoder complexities and delay constraints. Similar assumptions have been widely used in the literature of interference channels, see for example~\cite{Etkin}.

The service providers use power control to maintain certain throughput via controlling their transmission power levels on the downlinks. In interference channels, an increase in the transmission power on one of the downlinks causes interference on the other links and thus a degradation in the Signal to Interference and Noise Ratio (SINR) at the receiving sides of those links. In this regard,
let $\Lambda_{i,\max}$ denote the maximum transmission power allocated to service provider~$i$. A power control scheme of a service provider specifies the transmission power allocated to each end-user on the downlink.
Let the vector $\mbox{\boldmath$\lambda$}_i= (\lambda_{i,1}, \lambda_{i,2}, \cdots \lambda_{i,{|\mathcal{S}_i|}}) $ denote the power control scheme of service provider $i$. If all service providers transmit at their maximum power levels, the SINR of end-user $k\in\mathcal{S}_i$ can be represented by
\[
\gamma_{i,k}(\mbox{\boldmath$\lambda$}_i)= \frac{\lambda_{i,k}h_{i,k}}{\sigma_k + h_{i,k}(\Lambda_{i,\max} - \lambda_{i,k}) + \sum_{j\neq i}h_{j,k}\Lambda_{j,\max}},
\]
where $h_{i,k}$ is the path gain between the base station of service provider $i$ and end-user $k$,\footnote{In fact, it is path attenuation since $h_{i,k}<1$. In practice,  path attenuations may be obtained
by use of pilot signals or may not be explicitly discovered if the power control
mechanism is performed by adaptive dithering.} and $\sigma_k$ is the noise power at end-user $k$~\cite{Molisch}.
The achievable throughput rate at the downlink of user $k\in \mathcal{S}_i$
can be obtained using Shannon's formula
\beq
r_{i,k}(\mbox{\boldmath$\lambda$}_i)=W \log_2(1+\gamma_{i,k}(\mbox{\boldmath$\lambda$}_i)), \label{rate}
\eeq
and the aggregate rate on the downlink of service provider $i$ is thus given by
\beq
R_i(\mbox{\boldmath$\lambda$}_i)=\sum_{k\in \mathcal{S}_i}r_{i,k}(\mbox{\boldmath$\lambda$}_i). \label{rate_aggrg}
\eeq

It is plausible to measure the utilities of service providers from sharing the channel by the quality of service they provide on the downlinks. One important measure in this regard is the average delay of packet delivery, which can be reduced by improving the rate on the downlink. Let $\bar{R}_i$ be the long-term average downlink rate of service provider~$i$. We denote the utility of the service provider by the function $\mathcal{U}_i(.)$ which is strictly increasing  in $\bar{R}_i$.

A distinctive feature of secondary utilization of licensed spectrum bands is that the utility of secondary users, i.e., service providers, is discounted by some cost paid in the form of a fee to the license holder. See for example~\cite{Daoud},~\cite{Sarkar},~\cite{Mutlu} for studies that involve economic models and pricing techniques for secondary spectrum utilization. Here, we consider a pricing scheme where the license holder charges the service providers on usage basis per unit data transmitted on the downlinks. Let $c_i$ denote the price charged to service provider~$i$ per unit data transmitted on the channel. Thus, the optimal aggregate rate $\bar{R}^*_i$ of the service provider is a solution of the optimization problem
\begin{equation*}
\max_{\bar{R}_i}~~
\mathcal{U}_i(\bar{R}_i) - c_i\bar{R}_i,
\end{equation*}
which has a unique solution if $\mathcal{U}_i$ is concave.

From the standpoint of service provider~$i$, achieving $\bar{R}^*_i$ requires the service provider to transmit at a certain power level taking into consideration the interference created by other service providers. In the light of lack of central coordination, some service providers may unpredictably change their transmission power levels to adapt their rates according to their demand, thus causing variations in the interference to the other service providers. Sharing an interference channel with users that transmit at varying power levels can be modeled as a non-cooperative iterated game, where it can be assumed that the channel is offered to the service providers in rounds. In each round, the service providers choose their transmission power levels, which can vary from round to round according to their anticipated demand levels. Similar models for sharing interference channels have been considered, for example, in~\cite{Etkin}.

We refer to this game as the iterated power control game. In each round of the game, the service providers adapt their power levels based on some history of the game to maintain $\bar{R}_i^*$ regardless of the power control strategy of the opponents. The theory of zero-determinant strategies presented in Section~\ref{sec:theory} helps provide guidelines for power control in such environments that involve uncoordinated spectrum access. In the sequel, we present the iterated power control game for the case of two service providers, where the service providers fix their power allocation schemes $\mbox{\boldmath$\lambda$}_i$, but take binary decisions in each round on whether or not to engage their users. We identify the range of values of $\bar{R}_i^*$ that can be achieved and characterize strategies for achieving these values.



\subsection{Iterated Power Control Game with Two Service Providers and Binary Action Space}
Consider an iterated power control game with two service providers labeled $1$ and $2$. The payoff matrix of the single shot game is shown in Figure~\ref{fig:game}. For ease of exposition, we will assume that in each round of the game the service providers choose either to transmit at a certain power level or at zero power. That is, the service providers can choose between two actions: either to access or not to access the channel. We assume that if one service provider accesses the channel, it achieves downlink rate $R_1$ if it is service provider $1$ and rate $R_2$ if it is service provider $2$. Both $R_1$ and $R_2$ are described by~(\ref{rate_aggrg}). If both service providers access the channel, then potentially both achieve lower rates $\theta_i R_i$, $i=1,2$ with $0 < \theta_i < 1$. Trivially, the service provider that does not access the channel in any given round  achieves no rate.\footnote{Here, we assume a stationary model of user path gains. Namely, mobility of users may result in time-varying path gains as when a user enters a deep fade region or when handed off to a neighboring base-station of the same provider or to another provider in the same region. To handle such transients, the power control iterations need to be performed on a faster time-scale.}
However, the game can be extended to include power values that are not necessarily limited to $0$ and $\Lambda_{i,\max}$. Furthermore, extensions to more than two service providers follow directly from the generalization in Proposition~\ref{prop:multi}.


\begin{figure}[]
\centering
\begin{tabular}{c|c|c}
\backslashbox{provider 1}{provider 2}
& Access & No Access \\ \hline
Access &  $(\theta_1R_1,\theta_2R_2)$& $(R_1,0)$\\ \hline
No Access &  $(0,R_2)$ &  $(0,0)$
\end{tabular}
\caption{Payoff matrix of the power control game. Service providers~$1$ and~$2$ achieve rates $R_1$ and $R_2$, respectively, if they access the channel solely. If both providers access the channel simultaneously, then provider $i$ achieves $\theta_iR_i$.}
\label{fig:game}
\end{figure}

The payoff matrix of this game has a structure identified by Theorem~\ref{thm:fixing}. It allows any of the service providers to exert control on its rate. Specifically, let $1$ denote \emph{access} and $2$ denote \emph{no access}.
In any round, from the perspective of service provider~$i$, the game can be in one of four possible states given by the set
\[
\Omega_i=\{(1,1), (1,2), (2,1), (2,2)\},
\]
where the first element of a tuple refers to an action by service provider~$i$ and the second element refers to an action by the other service provider.
Let
$n_i(t) \in \{1, 2\}$ denote an action by service provider $i$ at round $t$ of the game. Also let ${\bf n}(t)=(n_1,n_2)$ and let
\[
p^{{\bf k}}_i = {\sf Pr}\left(n_i(t+1)=1|{\bf n}(t) ={\bf k}\right),~~\forall {\bf k}\in\Omega_i.
\]
Therefore, following the results in Theorem~\ref{thm:fixing}, service provider~$i$ can fix its long-term rate, $\bar{R}_i$, at any value
 in the interval $(0, \theta_iR_i]$
by accessing the channel in each round of the game according to the
following policy
\begin{eqnarray}
p^{1,1}_i &=& 1 + (1 -\theta_i R_i/\bar{R}_i)b_i, \label{eq.share1} \\
p^{1,2}_i &=& 1 + (1 -R_i/\bar{R}_i)b_i, \label{eq.share2} \\
p^{2,1}_i &=& b_i, \label{eq.share3} \\
p^{2,2}_i &=& b_i, \label{eq.share4}
\end{eqnarray}
where $b_i$ is chosen such that
\[
0  < b_i \leq \frac{1}{|1-R_i/\bar{R}_i|}.
\]

Obtaining values of $R_1, R_2, \theta_1, \theta_2$, which identify the range of possible fixations of the outcome of the game and the associated access strategies, hinges on the underlying power allocation scheme, $\mbox{\boldmath$\lambda$}_i$, applied by the service providers on the downlinks. In the following, we derive lumped parameters for computing these values for the max-min power allocation scheme that maximizes the minimum rate on the downlinks. Specifically, for service provider $i$, the max-min scheme requires solving the following optimization problem
\begin{equation*}
\begin{aligned}
&\underset{\mbox{\boldmath$\lambda$}_i}\max~\underset{k\in \mathcal{S}_i}\min~~
& & \gamma_{i,k}(\mbox{\boldmath$\lambda$}_i) \\
& \text{subject to}
& & \sum_{k\in \mathcal{S}_i} \lambda_{i,k}=\Lambda_{i,\max},
\end{aligned}
\end{equation*}
where it is implied that the service providers transmit at the maximum allowed power $\Lambda_{i,\max}$. A solution of this problem results in equal rates $r_i$ on all the downlinks of service provider~$i$.

To compute $R_i$, consider the case where only service provider~$i$ accesses the channel. In this case, the rate achieved at any of the downlinks $k \in \mathcal{S}_i$ is given by
\[
r_i =W\log_2\left(1+\frac{\lambda_{i,k}h_{i,k}}{\sigma_k + h_{i,k}(\Lambda_{i,\max} - \lambda_{i,k}) }\right),
\] and thus $R_i = |\mathcal{S}_i|r_i$.
Equal rates can be maintained on all the downlinks by choosing $\lambda_{i,k}$ and $\lambda_{i,l}$ for all $k,l \in \mathcal{S}_i$ such that
\[
\frac{h_{i,k} \lambda_{i,k}}{\sigma_k + h_{i,k}(\Lambda_{i,\max} - \lambda_{i,k})} = \frac{h_{i,l} \lambda_{i,l}}{\sigma_l + h_{i,l}(\Lambda_{i,\max} - \lambda_{i,l})},
\]
which can be equivalently written as
\[
\frac{h_{i,k} \lambda_{i,k}}{\sigma_k + h_{i,k}\Lambda_{i,\max}} = \frac{h_{i,l} \lambda_{i,l}}{\sigma_l + h_{i,l}\Lambda_{i,\max}} =:K,~~~\forall k,l \in \mathcal{S}_i.
\]
Thus,
\[
\lambda_{i,k} = K\left(\frac{\sigma_k}{h_{i,k}}+\Lambda_{i,\max}\right), ~~~\forall k \in \mathcal{S}_i.
\]
Note that $\Lambda_{i,\max}=\sum_{k\in \mathcal{S}_i}\lambda_{i,k}$
and therefore, for maximum $R_i$,
\[
K = \frac{\Lambda_{i,\max}}{\sum_{k\in\mathcal{S}_i}\frac{\sigma_k}{h_{i,k}} +
|\mathcal{S}_i|\Lambda_{i,\max}}.\]

$\theta_i$ can be computed in a similar fashion by considering both service providers $i$ and~$j$ simultaneously transmitting on the channel and taking into consideration the interference they create to each other. In this case, let $\tilde{r}_i$ denote the rate on each downlink of service provider $i$. Thus,
$\theta_iR_i = |\mathcal{S}_i|\tilde{r}_i$ where
for all $k\in\mathcal{S}_i$,
\begin{eqnarray*}
\tilde{r}_i =
W\log_2\left(1+\frac{\lambda_{i,k}h_{i,k}}{\sigma_k + h_{i,k}(\Lambda_{i,\max} - \lambda_{i,k}) + h_{j,k}\Lambda_{j,\max}}\right).
\end{eqnarray*}
The power distribution on the downlinks can be obtained by equalizing all the rates. Thus,
\[
\lambda_{i,k} = \tilde{K}\left(\frac{\sigma_k}{h_{i,k}}+\Lambda_{i,\max} +
\Lambda_{j,\max}\frac{h_{j,k}}{h_{i,k}} \right), ~~~\forall k \in \mathcal{S}_i,\]
where
\[
\tilde{K} = \frac{\Lambda_{i,\max}}
{\sum_{k\in\mathcal{S}_i} \frac{\sigma_k+\Lambda_{j,\max}h_{j,k}}{h_{i,k}}
+ |\mathcal{S}_i|\Lambda_{i,\max}
}.
\]



\section{Numerical Study} \label{sec:numerics}
In this section, we provide insight into the zero-determinant strategies for the $2\times 2$ game described in Figure~\ref{fig:game}. The structure of these strategies is given by formulae~(\ref{eq.share1})-(\ref{eq.share4}). Without loss of generality, we consider a symmetric game with $R_1=R_2=1.0$ and $\theta_1=\theta_2=0.5$. Here, each service provider can fix $\bar{R}_i$ to values in the range $(0,0.5]$. From the standpoint of service provider 1, i.e., the row player, the zero-determinant strategies $(p^{1,1}, p^{1,2}, p^{2,1}, p^{2,2})$ for a given $\bar{R}_1$  have the following structure
\beq
\left( 1 + (1-\frac{0.5}{\bar{R}_1})b_1,~1 + (1- \frac{1}{\bar{R}_1})b_1 ,~b_1 ,~b_1 \right),~\label{strategy}
\eeq
where
\[
0  < b_1 \leq \frac{1}{|1-1/\bar{R}_i|}.
\]

\subsection{Convergence of the Zero-Determinant Strategies}
First, consider the deterministic strategy $(1,0,1,1)$ which corresponds to $b_1=1$ and which allows the service provider to achieve a long-term average rate $\bar{R}_1=0.5$. Assume the strategy is played against service provider 2 which accesses the channel in each round with probability equal to  $1/2$. Figure~\ref{fig.comparison} shows the average rate of service provider $1$ at different rounds of the game as it converges, in the long term, to the value~$0.5$. Convergence paths are also provided for the strategies $(2/3, 0, 1/3, 1/3)$ and $(5/9, 0, 1/9, 1/9)$ which lead, respectively, to $\bar{R}_1=0.25$ and $\bar{R}_1=0.1$.

\begin{figure}
\centerline{{\includegraphics[width=5in]{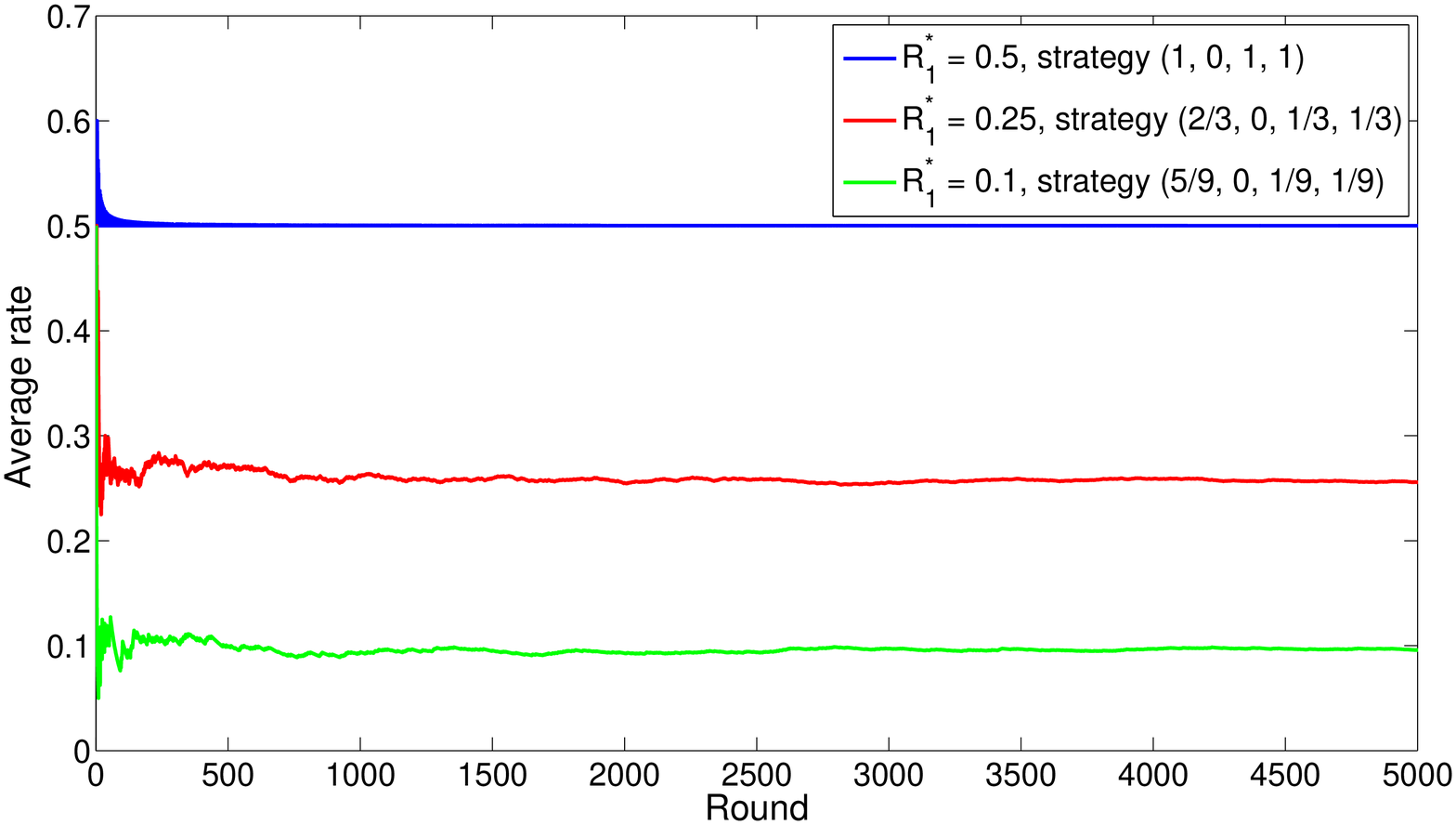}
\vspace*{-.25cm} }}\caption{Rate convergence under zero-determinants strategies in a $2\times2$ channel sharing game with $R_1=R_2=1.0$ and $\theta_1=\theta_2=0.5$. Service provider 2 (column player) uses a random strategy with probability of access = $1/2$ in each round.} \label{fig.comparison}
\end{figure}

A common factor of these strategies is that $p^{1,2}=0$, which means that if the service provider ends up using the channel alone in any round, it will not access the channel in the next round. Strategies with this property can be obtained by setting $b_1$ at the highest possible value. These rectifying strategies guarantee that there will be no long time periods of deviation from $\bar{R}_1$ which explains why all the previous strategies converge relatively quickly to $\bar{R}_1$.


Figure~\ref{fig:convergence} shows convergence paths of different strategies that achieve $\bar{R}_1=0.5$ including the strategy (1,0,1,1). All the strategies are played against the same service provider that has a probability of channel access equal to $1/2$. Note that as $p^{1,2}$ increases, strategies take longer time to converge. This is due to the fact that if the service provider accesses the channel in one round, then as $p^{1,2}$ increases, it is more likely that the service provider will access the channel in the next round, and thus, it is more likely to deviate more from  $\bar{R}_1$. In the meantime, an increase in $p^{1,2}$ is accompanied with a decrease in $p^{2,1}$ and $p^{2,2}$ by formula~(\ref{strategy}). This leads to long periods of channel access that are followed by long periods of no channel access and thus, the strategy tends to converge relatively slowly.

\begin{figure}
\centerline{{\includegraphics[width=5in]{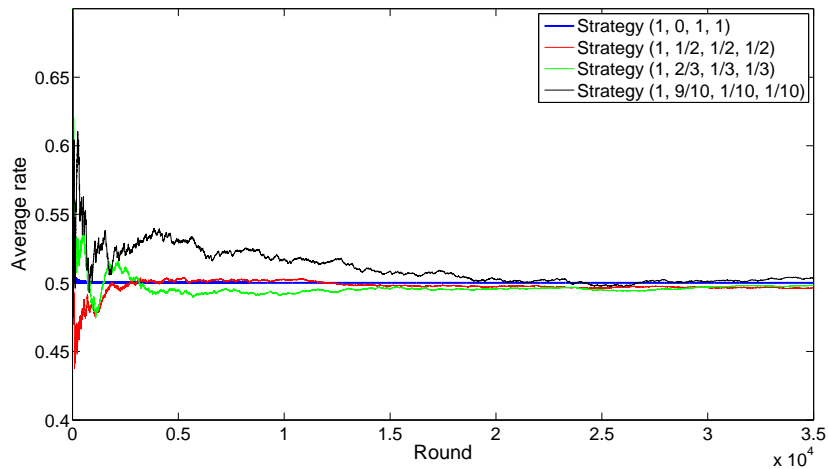}
\vspace*{-.25cm} }}\caption{Convergence paths of different zero-determinant strategies that can be applied by service provider 1 (row player) in the power control game. Strategies with higher $p^{1,2}$ are slower to converge, since as $p^{1,2}$ increases, it becomes more likely for the service provider to access the channel in the next round if already accessed the channel in the previous round. This leads to longer periods of deviation from $\bar{R}_i$, and thus, longer conversion times. } \label{fig:convergence}
\end{figure}

A less obvious conclusion can be carried from games of more than two service providers. For example, consider the game with three service providers $i,j,k$. Assume that in any round of the game, service provider~$i$ achieves one of the following rates:
\[
  \left\{\begin{array}{ll}
 R_i, &\mbox{if accesses the channel alone,} \\
 \alpha_1R_i, &\mbox{if accesses the channel with another provider,}\\
  \alpha_2R_i, &\mbox{if all the service providers access the channel,} \\
   0, &\mbox{otherwise,}
\end{array}\right.
\]
where $0 < \alpha_1, \alpha_2 < 1$. Let $x_i,x_j,x_k$ denote, respectively, the actions of player $i,j,k$ in any round, where  $x_i,x_j,x_k \in \{1,2\}$ and such that $1$ implies access and $2$ implies no access. Let $p^{x_i,x_j,x_k}_i$ denote the probability that service provider~$i$ accesses the channel if the state of the game was $(x_i,x_j,x_k)$ in the previous round.

Following Proposition~\ref{prop:multi}, a zero-determinant strategy allows service provider~$i$ to fix $\bar{R}_i$ at any value in the interval $(0, \alpha_2R_i]$. The structure of the policy is given by:
\begin{eqnarray}
p^{1,1,1}_i &=& 1 + (1 -\frac{\alpha_2R_i}{\bar{R}_i})b_i, \nonumber \\
p^{1,1,2}_i &=& p^{1,2,1}_i = 1 + (1 -\frac{\alpha_1R_i}{\bar{R}_i})b_i, \nonumber \\
p^{1,2,2}_i &=& 1 + (1 - \frac{R_i}{\bar{R}_i})b_i, \nonumber \\
p^{2,1,1}_i &=& p^{2,1,2}_i = p^{2,2,1}_i = p^{2,2,2}_i = b_i, \nonumber
\end{eqnarray}
where
\[
0  < b_i \leq \frac{1}{|1-R_i/\bar{R}_i|}.
\]

Figure~\ref{fig.comparison_3players} shows convergence paths of different strategies when $R_i=1.0, \alpha_1=1/2$, and $\alpha_2=1/3$. All the strategies aim to fix $\bar{R}_i$ at the maximum possible value, $1/3$, where service providers~$j$ and~$k$ access the channel at each round with probability~$1/2$ and $3/4$, respectively. A strategy is displayed in the figure by an $8$-element tuple where the first $4$ elements correspond to $p^{1,1,1}_i, p^{1,1,2}_i, p^{1,2,1}_i$, and $p^{1,2,2}_i$, respectively. Note that, since $\bar{R}_i$ is fixed at the maximum value, then $p^{1,1,1}_1 = 1$ for all the strategies. The pattern observed in Figure~\ref{fig:convergence} applies to Figure~\ref{fig.comparison_3players} where strategies that converge quickly are the strategies that have lower $p^{1,1,2}_i$, $p^{1,2,1}_i$, and  $p^{1,2,2}_i$, i.e., these are the strategies that are less likely to access the channel if they achieved more than the targeted rate, $1/3$, in the previous round.


\begin{figure}
\centerline{{\includegraphics[width=5in]{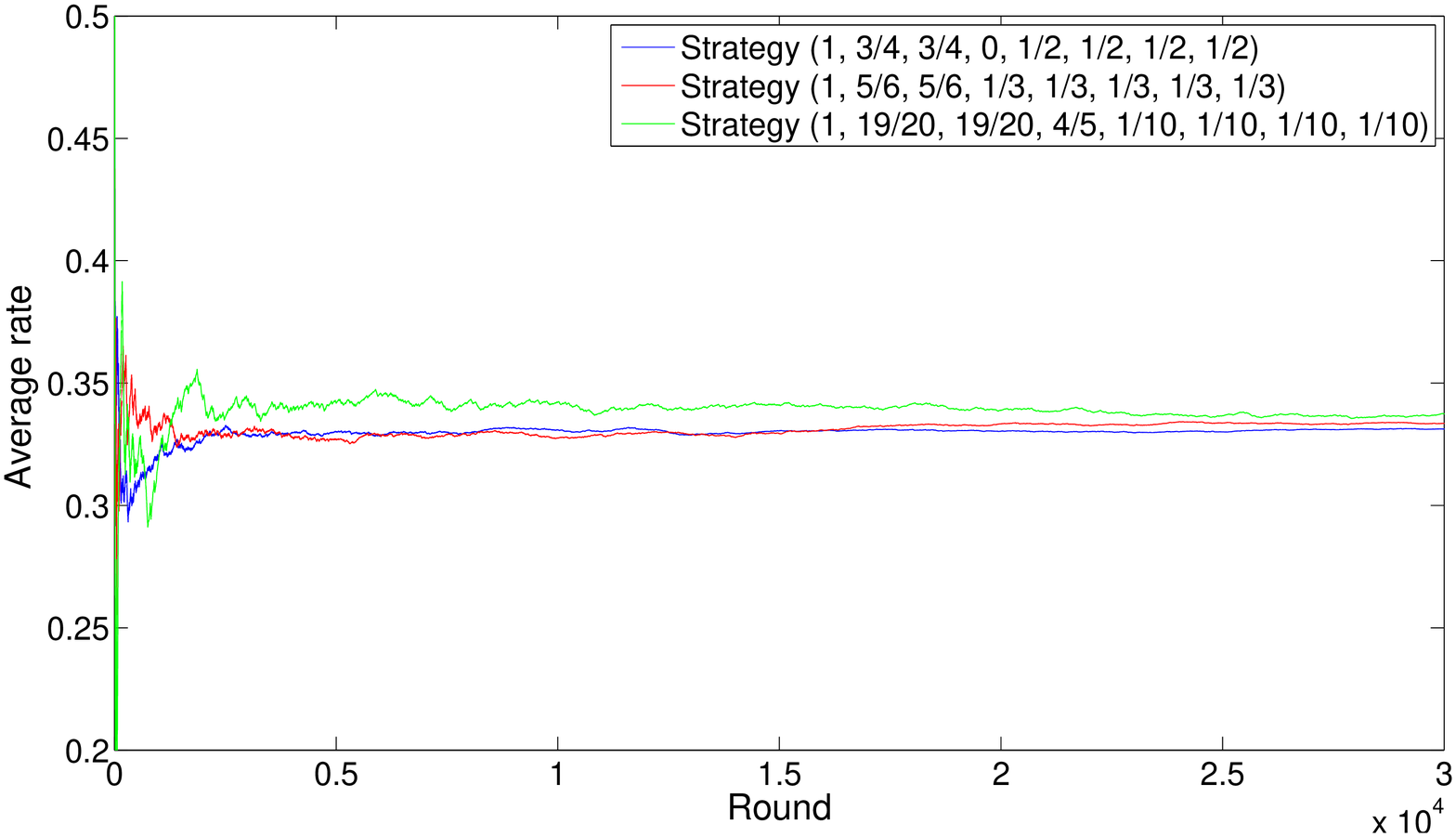}
\vspace*{-.25cm} }}\caption{Convergence paths for multiple zero-determinant strategies in the power control game with three service providers. Strategies that are more likely to rectify, if exceeding the targeted rate, are the strategies that converge relatively quicker. } \label{fig.comparison_3players}
\end{figure}

\subsection{Zero-Determinant Strategies and Power Consumption}
Next, we investigate the impact of the zero-determinant strategies on the average power consumption of the service providers. In the considered power control game, the service providers take binary decisions in each round whether or not to access the channel. If the channel is to be accessed, service provider~$i$ transmits at the maximum allowed power level $\Lambda_{i,\max}$. Therefore, power consumption over the course of the game has a probability distribution derived from the stationary distribution of the state of the game,  $\mbox{\boldmath${\pi}$}$. In particular, consider a $2\times 2$ power control game and consider service provider~$1$, i.e., the row player. The average consumed power is given by
\[
\Lambda_{1,avg}= \Lambda_{1,\max}(\pi_{1,1} +\pi_{1,2}).
\]
Here, $\pi_{1,1}$  is the proportion of rounds where both service providers transmit and $\pi_{1,2}$ is the proportion where only service provider~$1$ transmits.


Consider the game in Figure~\ref{fig:game} and assume that $R_1 = R_2 = 1.0$ and $\theta_1 = \theta_2 = 0.5$. Assume that both service providers use zero-determinant strategies to achieve $\bar{R}_1=0.5$ and $\bar{R}_2= 0.25$. The impact of the different strategies on power savings is shown in Figure~\ref{fig:power}. The horizontal axis displays possible strategies of service provider~$1$ with each strategy denoted by a different value of the variable $b_1$ defined in~(\ref{strategy}). All the values of $b_1$ are taken from the feasible range $[0.1, 1]$, and a common factor of all these strategies is that $p^{1,1} = 1$. We show the proportion of rounds in which service provider~$1$ accesses the channel,  ($\pi_{1,1} + \pi_{1,2}$), where each curve corresponds to a different strategy of service provider~$2$. Here, the strategies of service provider~$2$ are denoted by the vector $\mathbf{q} = (q^{1,1}, q^{2,1}, q^{1,2}, q^{2,2})$, where, following the convention in Section~\ref{sec:theory}, $q^{x,y}$ is the probability that service provider $2$ will access the channel if service provider~$1$ played action~$x$ and service provider~$2$ played action~$y$ in the previous round.

The figure shows that power consumption of service provider~$1$ is unimodal in the value of $b_1$, but can be increasing or decreasing according to the strategy of the opponent. The figure also shows that there exists a trend in power savings that the service provider can achieve from playing against different opponent strategies. Namely, playing against strategies that have relatively low $q^{1,1}$  leads to more power savings. The intuition behind this observation is that when $q^{1,1}$ is low, service provider~$2$ is more likely to skip the channel in the next round if both service providers accessed the channel in the current round. Now since $p^{1,1} = 1$, service provider~$1$ will have the whole channel with probability~$1$ in the next round, i.e., transmitting with no interference and thus achieving high rate.

This argument can be solidified by looking at $p^{1,2}$ vs $q^{1,2}$ and $p^{2,1}$ vs $q^{2,1}$. Notice that by increasing $b_1$, $p^{1,2}$ decreases and $p^{2,1}$ increases, and thus, if playing against a strategy with relatively high  $q^{1,2}$ (meaning  low  $q^{2,1}$) such as $\mathbf{q} = (2/3,0,1/3,1/3)$, the gap between the previous values is going to increase. This means that, if in any round only one service providers accessed the channel, it is more likely for the other service provider to access the channel in the next round, and visa versa, leading to more power savings.  On the other hand, the gap decreases if compared to a strategy with relatively low  $q^{1,2}$ and high   $q^{2,1}$ such as $\mathbf{q} = (9/10,7/10,1/10,1/10)$. In such a case, power consumption is going to increase.



\begin{figure}
\centerline{{\includegraphics[width=5in]{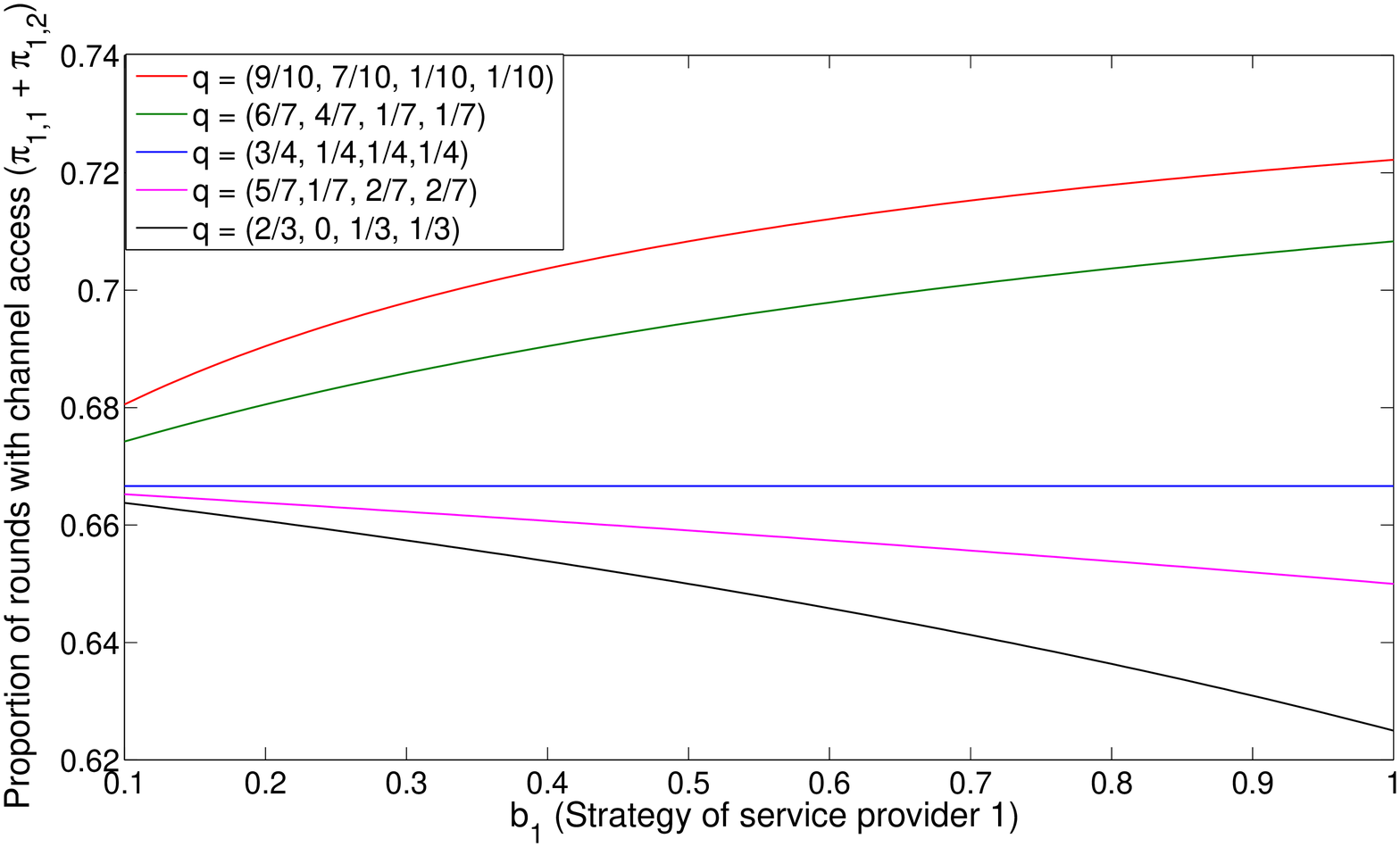}
\vspace*{-.25cm} }}\caption{Proportion of rounds in which service provider~$1$ (row player) accesses the channel displayed for different strategies of service provider~$2$. Playing against strategies that are more likely to skip the channel when the other service provider accesses the channel leads to power savings.} \label{fig:power}
\end{figure}

\section{Conclusion} \label{sec:concl}
In this paper, we considered private commons as a model for secondary sharing of licensed spectrum bands. The system involves multiple wireless service providers sharing an interference channel in uncoordinated fashion and servicing their own populations of co-located end-users. The problem of aggregate downlink power control is formulated as a non-cooperative iterated game. In this regard, we considered a set of Markovian strategies known as ``zero-determinant" strategies that are primarily developed for the iterated Prisoners' Dilemma game and which are shown to allow players to exert control on each other's score. We extended these strategies to be considered for any $2\times 2$ game and in a way based on the structure of the game. We showed that the spectrum sharing game admits an appealing structure that allows service providers to employ power control strategies to set their own aggregate rates regardless of the strategies of other service providers. We provided numerical experiments to study the convergence behavior of these strategies and their impact on power consumption.

\section*{Acknowledgment}
This work was funded by NSERC Strategic Project grant and NSF CNS grant 1116626.




\begin{thebibliography}{1}

\bibitem{Daoud}
A. Al Daoud, M. Alanyali, and D. Starobinski.
Pricing Strategies for Spectrum Lease in Secondary Markets.
{\em IEEE/ACM Transactions on Networking}, Vol.~18, No.~2,
 pp.~462--475, April 2010.

\bibitem{Alpcan_book} T. Alpcan, H. Boche, M.L. Honig, and H. Vincent Poor.
Mechanisms and Games for Dynamic Spectrum Allocation,
{\em Cambridge Press}, 2013.

\bibitem{Alpcan} T. Alpcan, T. Basar, R. Srikant, and Eitan Altman. CDMA
Uplink power control as a noncooperative game.
{\em in Proc. IEEE Conference on Decision and Control}, pp.~197--202, 2001.

\bibitem{Altman}
E. Altman, K. Avrachenkov, G. Miller, and B. Prabhu.
Discrete power control: Cooperative and non-cooperative optimization.
{\em in Proc. IEEE INFOCOM}, pp.~37--45, 2007.

\bibitem{Attar}   A. Attar, M.R. Nakhai, A.H. Aghvami.
Cognitive Radio game for secondary spectrum access problem.
{\em IEEE Transactions on Wireless Communications},
Vol.~8, No.~4, pp.~2121--2131, April 2009.

\bibitem{Aumann}  Robert Aumann and Adam Brandenburger.
Epistemic Conditions for Nash Equilibrium.
{\em Econometrica},
Vol.~63, No.~5, pp.~1161-1180, September 1995.

\bibitem{Axelrod} R. Axelrod. {\em The evolution of cooperation}.
Basic Books, New York, 1984.

\bibitem{Buddhikot} M.M.~Buddhikot.
Understanding dynamic specrum access: Models, taxonomy and challenges.
{\em in Proc. of IEEE Symposium on New Frontiers in Dynamic Spectrum Access Networks (DySPAN)}, pp.~649--663, 2007.

\bibitem{Boerlijst} M.C. Boerlijst, M.A. Nowak, and K. Sigmund.
Equal pay for all prisoners. {\em The American mathematical monthly,} Vol.~104,
No.~4, pp.~303--305, April 1997.

\bibitem{Cisco} Cisco Visual Networking Index.
Global mobile data traffic forecast
update, 2011 - 2016, http://tinyurl.com/VNI2012, May 2012.

\bibitem{Chung} S.T. Chung, S. Kim, J. Lee, and J.M. Cioffi.
A game-theoretic approach to power allocation in frequency-selective Gaussian interference channels.
{\em in Proc. of IEEE International Symposium on Information Thoery}, pp.~136--136, 2003.

\bibitem{Etkin} R. Etkin, A. Parekh and D. Tse.
Spectrum sharing for unlicensed bands.
{\em IEEE Journal on Selected Areas in Communications},
Vol.~25, No.~3, pp.~517--528, April 2007.

\bibitem{FCC} Federal Communications Commission.
Promoting efficient use of spectrum through elimination of Barriers to the development of secondary markets.
{\em Second Report and Order on Reconsideration and Second Further Further Notice of Proposed Rule Making}, 2004.

\bibitem{Felegyhazi}    M. Felegyhazi, J.P. Hubaux, and L. Buttyan.
Nash equilibria of packet forwarding strategies in wireless ad hoc networks.
{\em  IEEE Transactions on Mobile Computing}, Vol.~5, No.~5, pp.~463--476,
May 2006.

\bibitem{Huang} J. Huang,  R.A. Berry, and M.L. Honig. Distributed interference compensation for wireless networks.
{\em IEEE Journal on Selected Areas of Communications},
Vol.~24, No.~5, pp.~1074--1084, May 2006.

\bibitem{Jaramillo}  J.J. Jaramillo and R. Srikant.
DARWIN: distributed and adaptive reputation mechanism for wireless ad-hoc networks.
{\em in Proc. of ACM International Conference on Mobile Computing and Networking (Mobicom)}, 2007.

\bibitem{Jin} Y. Jin and G. Kesidis.
Distributed Contention Window Control for Selfish Users in IEEE 802.11 Wireless LANs.
{\em IEEE Journal on Selected Areas in Communications},.
Vol.~25, No.~6, pp.~1113--1123,  August 2007.

\bibitem{Jin2}
Y. Jin and G. Kesidis.
A channel-aware MAC protocol in an ALOHA network with selfish users.
{\em IEEE Journal on Selected Areas in Communications- Special Issue on Game Theory in Wireless Communications},
Vol.~30, No.~1, pp.~128--137, January 2012.

\bibitem{Korilis} Y.A. Korilis,  A.A. Lazar, A. Orda.
Architecting noncooperative networks.
{\em IEEE Journal on Selected Areas in Communications},
Vol. 13, No.7, pp.1241-1251, September 1995.

\bibitem{Molisch} A.F. Molisch, Wireless Communications, {\em Wiley}, 2010.

\bibitem{Sarkar} P.K. Muthuswamy, K. Kar, A. Gupta, S. Sarkar, and G. Kasbekar. Portfolio Optimization in Secondary Spectrum Markets.
In {\em Proc. WiOpt}, pp.~249--256, 2011.

\bibitem{Mutlu}
H. Mutlu, M. Alanyali, and D. Starobinski.
Spot Pricing of Secondary Spectrum Access in Wireless Cellular Networks.
{\em IEEE/ACM Transactions on Networking}, Vol.~17, No.~6,
 pp.~1794--1804, December 2009.

\bibitem{Niyato} D. Niyato and E. Hossain. Competitive pricing for spectrum sharing in cognitive radio
networks: dynamic game, inefficiency of Nash equilibrium, and collusion. {\em IEEE Journal
on Selected Areas in Communications}, Vol.~26, No.~17, pp.~192–-202, January 2008.

\bibitem{Osborne} M. J. Osborne and A. Rubinstein. {\em A Course in Game Theory}. MIT Press Books, 1999.

\bibitem{Press} W. H. Press and F. J. Dyson. Iterated Prisoner's Dilemma contains strategies that dominate any evolutionary opponent. {\em Proceedings of the National Academy of Sciences} Vol.~109, No.~26, pp.~10409--10413, June 2012.

\bibitem{RSPG} Radio Spectrum Policy Group. {\em Report on collective use of spectrum and other sharing approaches}. 2011.

\bibitem{Mandayam}
C.U. Saraydar,  N.B. Mandayam and D.  Goodman.
Efficient power control via pricing in wireless data networks.
{\em IEEE Transactions on Communications},
Vol. 50,  No. 2, pp.  291-303,
February 2002.


\bibitem{Wang} B. Wang, Y. Wu, K.J. Ray Liu.
Game theory for cognitive radio networks: An overview.
{\em Computer Networks}, Vol.~54, No.~14, pp.~2537--2561, October 2010.

\end{thebibliography}
%

\end{document}